\documentclass[a4paper,UKenglish,cleveref, autoref, thm-restate]{lipics-v2021}
\nolinenumbers

\hideLIPIcs

\usepackage[utf8]{inputenc}

\usepackage{amsthm, microtype,hyperref,graphicx}
   \hypersetup{%
      breaklinks,%
      colorlinks=true,%
      urlcolor=[rgb]{0.25,0.0,0.0},%
      linkcolor=[rgb]{0.5,0.0,0.0},%
      citecolor=[rgb]{0,0.2,0.445},%
      filecolor=[rgb]{0,0,0.4},
      anchorcolor=[rgb]={0.0,0.1,0.2}%
   }
   
\usepackage{makecell}

\newcommand{\Jejcic}{{Jej\v{c}i\v{c}}}

\newcommand{\DT}{\mathrm{DT}}
\newcommand{\R}{\mathbb{R}}
\newcommand{\inside}{\textsc{inside}}
\newcommand{\outside}{\textsc{outside}}
\renewcommand{\path}{\textsc{path}}
\newcommand{\vertex}{\textsc{vertex}}
\newcommand{\triangulation}{\textsc{tri}}
\newcommand{\DP}{\Delta}
\newcommand{\Tree}{\mathcal{T}}
\newcommand{\eps}{\varepsilon}
\newcommand{\ray}{\vec{r}}

\title{Shortest Path Separators in Unit Disk Graphs}

\author{Elfarouk Harb}{Department of Computer Science, University of Illinois at Urbana-Champaign, USA}{eyharb2@illinois.edu}{https://orcid.org/0000-0002-7770-0932}{}
\author{Zhengcheng Huang}{Department of Computer Science, University of Illinois at Urbana-Champaign, USA}{zh3@illinois.edu}{https://orcid.org/0009-0006-0974-8818}{}
\author{Da Wei Zheng}{Department of Computer Science, University of Illinois at Urbana-Champaign, USA}{dwzheng2@illinois.edu}{https://orcid.org/0000-0002-0844-9457}{}

\titlerunning{Shortest Path Separators in Unit Disk Graphs}
\authorrunning{E. Harb, Z. Huang and D.\,W. Zheng}

\Copyright{Elfarouk Harb, Zhengcheng Huang and Da Wei Zheng}



\EventEditors{Timothy Chan, Johannes Fischer, John Iacono, and Grzegorz Herman}
\EventNoEds{4}
\EventLongTitle{32nd Annual European Symposium on Algorithms (ESA 2024)}
\EventShortTitle{ESA 2024}
\EventAcronym{ESA}
\EventYear{2024}
\EventDate{September 2--4, 2024}
\EventLocation{Royal Holloway, London, United Kingdom}
\EventLogo{}
\SeriesVolume{308}
\ArticleNo{80}

\begin{document}

\maketitle

\begin{abstract}
We introduce a new balanced separator theorem for unit-disk graphs involving two shortest paths combined with the $1$-hop neighbours of those paths and two other vertices.
This answers an open problem of Yan, Xiang and Dragan [CGTA '12] and improves their result that requires removing the $3$-hop neighbourhood of two shortest paths. 
Our proof uses very different ideas, including Delaunay triangulations and a generalization of the celebrated balanced separator theorem of Lipton and Tarjan [J. Appl. Math. '79] to systems of non-intersecting paths.
\end{abstract}

\section{Introduction}
A \emph{geometric intersection graph} is an undirected graph where each vertex corresponds to a geometric object, and edges indicate which pairs of objects intersect each other. One common type of geometric intersection graph is the unit disk graph, which are geometric intersection graphs of disks with diameter $1$. These graphs arise in modeling wireless communication. Such graphs also appear in applications such as VLSI design. These graphs have been extensively studied in the computational geometry community.

There have been many papers studying different optimization problems in unit disk graphs~\cite{ClarkCJ90, YanXD12, KaplanMRS18, ChanS19a, ChangGL24}, as they are one of the simplest types of geometric intersection graph.
As such, results for unit disk graphs are often used as a starting point to generalize to more general geometric intersection graphs with more complicated objects~\cite{HuntMRRRS98, ErlebachJS05, ChanS19, Har-PeledY22, Chan23}.
Over the years, researchers have built up a toolbox of techniques that can be applied to unit disk graphs. Some of the techniques have been inspired by the more geometric aspects of these graphs, such as geometric data structures~\cite{ChanS19, Chan23}, planar spanners \cite{ChanS19a}, and well separated pair decompositions~\cite{GaoZ05}.
Other techniques have been inspired by planar graph tools, such as planar separators~\cite{MillerTTV97, BergKMT23}, shortest path seperators~\cite{YanXD12}, and more recently VC-dimension in planar graphs~\cite{LiP19,LeW21,ChangGL24}. 
Thus, there is much value in deepening our understanding of unit disk graphs, and in particular further building up its toolbox, as we do in this paper.

\subparagraph*{Separator theorems}
In graphs, a \emph{separator} is a small set of vertices whose removal splits the graph into smaller components. Separators are very useful for designing divide-and-conquer algorithms.
Planar graphs are well-known for admitting good separators. The first separator theorem for planar graphs was due to Lipton and Tarjan~\cite{LiptonT1979}, who proved that every planar graph on $n$ vertices admits a separator of size $O(\sqrt{n})$ that can be computed in $O(n)$ time. Since then, many variants of separator theorems have been proven for planar graphs~\cite{LiptonT1979, Miller86, Frederickson87, Goodrich95, Thorup04, KleinMS13, ChangKT22}. Some of these results can be naturally extended to graphs with bounded genus~\cite{djidjev1981separator,GilbertHT84} or to minor-free graphs~\cite{AlonSR90, KawarabayashiR10, Wulff-Nilsen11}.

For unit disk graphs, many different separators exist, such as line separators~\cite{CarmiCKKORRSS20} and clique separators~\cite{BergKMT23}. 
When the unit disks have low \emph{ply}\footnote{the number of disks intersecting any given point}, good separators are also known to exist~\cite{MillerTTV97, SmithW98}.
Separators are also known for more general objects, such as fat objects~\cite{Chan03} or low-density objects~\cite{Har-PeledQ17}.
Intersection graphs induced by arbitrary curves in the plane, also known as string graphs, have been studied in a series of works~\cite{FoxPach08, FoxP14, Matousek14, Lee17}. Remarkably, Lee~\cite{Lee17} proved that every string graph having $m$ edges admits a balanced separator with $O(\sqrt{m})$ vertices. 

\subparagraph*{Shortest path separators} The usefulness of a separator does not necessarily only depend on the number of vertices. An excellent example of good separators with large size are {\em shortest path separators}, {\em i.e.} separators consisting of a constant number of shortest paths. Thorup showed that every planar graph admits a separator consisting of at most two shortest paths~\cite{Thorup04}. Similar results have been proven for minor-free graphs~\cite{AbrahamG06}. Shortest path separators have been used extensively in distance-related problems in planar graphs, such as distance oracles~\cite{Thorup04,Klein02}, planar emulators~\cite{ChangKT22}, and even network design problems~\cite{KawarabayashiS21,FriggstadM23}.

A natural question  to ask is whether shortest path separator theorems can be adapted to unit disk graphs.
Naively, such separators cannot exist, as the clique on $n$ vertices is realizable as a unit disk graph for which no such separator can exist.
However, we can strengthen the separator by also removing vertices in the $k$-neighborhood of the shortest path, \emph{i.e.} vertices that are at a distance of at most $k$ from the shortest path.
Yan, Xiang and Dragan~\cite{YanXD12} proved that every unit disk graph admits a shortest path $3$-neighborhood separator, that is, by removing two shortest paths and all vertices in the $3$-hop neighborhood of any vertex on the shortest paths, the remaining graph is disconnected with every component having size at most $2/3$ of the vertices of the original graph.
They left open the question of whether there exists a shortest path $1$-neighborhood separator.

\subparagraph*{Our results} 
We answer the open question of Yan, Xiang and Dragan~\cite{YanXD12} in the affirmative.
We show that every unit disk graph has a $1$-neighborhood separator. In particular, it suffices to only remove the $1$-neighborhood of two shortest paths plus the $1$-neighborhood of two other vertices. 
While the proof of Yan, Xiang, and Dragan manipulates crossings in the intersection graph, our proof uses very different ideas involving paths in Delaunay triangulations and a generalization of the shortest path separators of Lipton and Tarjan to sets of \emph{weakly non-crossing paths} in a triangulated planar graph that may be of independent interest. This result is optimal since as mentioned earlier, a shortest path 0-neighborhood separator does not exist for a clique, which is realizable as a unit disk graph.

\begin{theorem}
    \label{thm:sp-sep}
    Every unit disk graph admits a shortest path 1-neighborhood separator.
\end{theorem}

\subparagraph*{A first attempt}
\begin{figure}
    \centering
    \includegraphics[page=1,height=0.22\textheight]{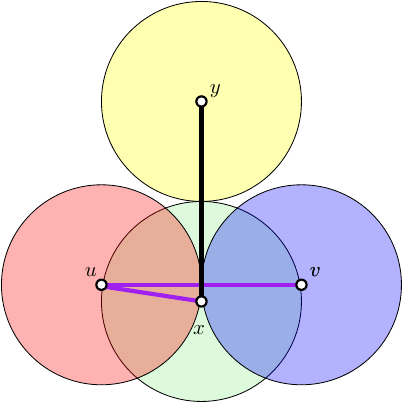}
    \includegraphics[page=2,height=0.2\textheight]{figures/counter_example.pdf}
    \caption{
    \textbf{(Left)} The points $u, v, x, y\in S$ are drawn with circles of radius $1/2$.
    The unique shortest path tree in $G$ with starting vertex $u$ has a crossing edge.
    \textbf{(Right)} Two reflected copies results in a graph where no plane shortest path tree exists, regardless of starting vertex.
    }
    \label{fig:counterexample}
\end{figure}

To illustrate the difficulty and to fully appreciate our contribution, let us consider one approach to construct shortest path 1-neighborhood separators for a unit disk graph $G = (V, E)$.
To do so, we will make two overly wishful assumptions (that would be great if they always held). 

Let $\Tree$ be a shortest path tree of $G$ starting at a fixed vertex $s \in V$. 
We will wishfully assume that the induced drawing of $\Tree$ in the plane has no crossings (assumption 1). 
Next, we will assume that we can triangulate $\Tree$ to get a graph $G_\Tree$ (assumption 2) such that every edge of the triangulated graph $G_\Tree$ is an edge in $G$. 
Now, we can use the shortest path separator theorem of Lipton and Tarjan~\cite{LiptonT1979}, on $G_\Tree$ with spanning tree $\Tree$ to get a Jordan curve $C$ that is a separator for $G_\Tree$.
We can prove (and do so in \Cref{lem:crossinginterior}) that all edges $uv\in E$ have the property that for all other edges crossing the line segment between $u$ and $v$ the crossing edge has at least one end point adjacent to either $u$ or $v$ 
(we call this property \emph{cross-dominating} or \emph{crominating}%
\footnote{This is a perfectly cromulent portmanteau to use.}
for short).
Thus the cycle $C$ is in fact also a shortest path $1$-neighborhood separator of $G$.

\begin{enumerate}
\item 
Our first assumption was that we could find a shortest path tree $\Tree$ of $G$ whose natural embedding has no crossings.
This is not always the case, there are examples (see \Cref{fig:counterexample}) of unit disk graphs $G$ where no such shortest path tree exists.
Instead, we will construct a non-crossing path system $\Pi$ consisting of \emph{pseudo-shortest paths}, i.e., for every vertex $u\in V$ we will find a path $\Pi[u]$ to $s$ such that $\Pi[u]$ consists only of vertices on the shortest path between $u$ and $s$ and $1$-neighbors of the shortest path. 
We show an extension of the planar separator algorithm to find a balanced separator in \emph{path systems} of planar graphs.

\item 
Our second assumption was that we could triangulate the tree $\Tree$ to get a graph $G_\Tree$ such that every edge of the triangulation is an edge in $G$. 
We instead prove that all edges of the Delaunay triangulation have the crominating property, and furthermore, we show that we can construct $\Pi$ using only edges of the Delaunay triangulation of the centers of the disk.%
\footnote{
In fact, it is possible to show that there exists a triangulation (specifically the \emph{edge constrained Delaunay triangulation} of Chew~\cite{Chew89}) of any plane tree such that all edges are crominating. However, we will not discuss this further since the edge constrained Delaunay triangulation of a spanning subset of edges of a Delaunay triangulation is the original Delaunay triangulation.
}
\end{enumerate}


\section{Preliminaries}

\subparagraph*{Crossings.\footnote{Our definitions of {\em crossing} are special cases of the definition from~\cite{EricksonN11}.}} For two line segments $uv$ and $pq$, we say that $uv$ \emph{crosses} $pq$ if there is a point on both line segments that is not an endpoint of either $uv$ or $pq$. 
Given two simple polygonal chains $P = [p_0, p_1, \dots, p_\ell]$ and $Q = [q_0, q_1, \dots, q_\ell]$, we say that $P$ and $Q$ have a \emph{forward crossing} if the following hold:

\begin{figure}
    \centering
    \includegraphics{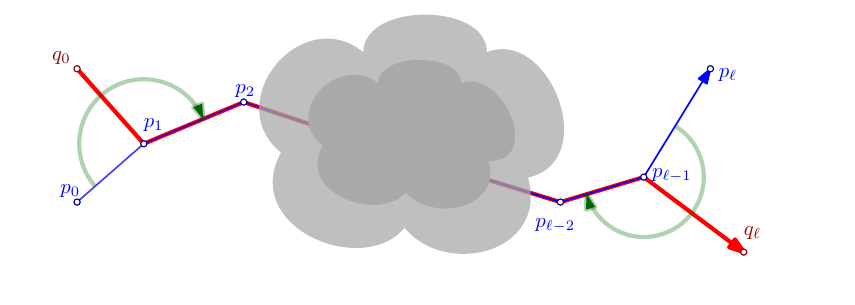}
    \caption{An example of a forward crossing between the red and blue chain. The cloud obscures the shared middle part of the polygonal chains}
    \label{fig:forward_crossing}
\end{figure}

\begin{enumerate}
    \item $p_i = q_i$ for all $1\le i \le \ell-1$
    \item The cyclic order of $p_0, q_0, p_2$ around $p_1$ is the same as the cyclic ordering of $p_\ell, q_\ell, p_{\ell-2}$ around $p_{\ell-1}$.
\end{enumerate}

See \Cref{fig:forward_crossing} for an example of a forward crossing.
The polygonal chains have a \emph{backwards crossing} if $P$ and the reversal of $Q$ have a forward crossing. 
Finally, we say that $P$ and $Q$ {\em cross} if subpaths of $P$ and $Q$ have a forward crossing, a backwards crossing, or two edges $p_ip_{i+1}$ and $q_jq_{j+1}$ that cross.
Equivalently $P$ and $Q$ do not cross if and only if there exists a small local perturbation of the vertices of each polygonal chain such that the two chains share no point in common (this is true as we only consider paths without spurs or forks~\cite[Section 3.1]{ChangEX15}).
If a collection $\Pi$ of polygonal chains are non-crossing (i.e. no pair of paths $P, Q\in \Pi$ cross), then we can also perturb the vertices of all the chains so that no two chains intersect (even at end points).
We note that all chains that we will discuss are fairly nice, and amenable to these two characterizations of crossing.
We encourage the interested reader to see the discussion in Chang, Erickson, and Xu~\cite{ChangEX15} for a thoroughly comprehensive discussion of crossing and non-crossing polygonal chains.

Let $G$ be a given graph with a straight-line embedding in $\R^2$. Then each edge of $G$ embeds onto a line segment and each path in $G$ embeds onto a polygonal chain. We say that two edges of $G$ cross if their corresponding line segments cross. Similarly, two paths in $G$ cross if their corresponding polygonal chains cross.

\subparagraph*{Shortest path separators.} Lipton and Tarjan~\cite{LiptonT1979} showed that a connected triangulated planar graph on $n$ vertices with arbitrary weights on the vertices has a \emph{balanced shortest path separator} which is a Jordan curve consisting of two shortest paths and one edge such that the interior and exterior of the curve each have at most $2/3$ of the total weight of the vertices of the graph.

\begin{theorem}[Balanced shortest path separator of Lipton-Tarjan~\cite{LiptonT1979}]
\label{thm:LTseparator}
Let $G = (V, E)$ be a maximally triangulated planar graph with $n$ vertices and let $T$ be a spanning shortest path tree of $G$ rooted at $s$. Suppose that each vertex of the graph has a weight, and let $W$ be the total weight of all the vertices. Then there exists an edge $uv \in E$ such that the Jordan curve $C$ defined by the edge $uv$ along with the path between $u$ and $v$ in $T$ separates the graph into an interior $A$ and exterior $B$ that each have weight at most $2W/3$.
\end{theorem}

\subparagraph*{Path systems.} 
Given a graph $G = (V,E)$ and a fixed vertex $s \in V$, we define a {\em path system to $s$} as a function $\Pi$ defined on some vertex set $V'\subseteq V$ that maps each vertex $u\in V'$ to a directed path $\Pi[u]$ from $u$ to $s$ in $G$. For this paper, we will always assume that our paths are \emph{simple}, that is no vertex is visited more than once on the path. 
We will abuse notation and use $\Pi$ to refer to the collection of paths to $s$.
When $G$ is planar, we say that the path system is {\em non-crossing} if for every pair of vertices $u,v\in V(G)$, the paths $\Pi[u]$ and $\Pi[v]$ are non-crossing. 
A path system is a {\em spanning} path system if every vertex in $V$ is on at least one path of $\Pi$.

\subparagraph*{Pseudo-shortest paths.} 
Given a graph $G$, a {\em pseudo-shortest path} $P$ from $u$ to $v$ for two vertices $u,v \in V$ is a path that starts at $u$, ends at $v$, contains all vertices of a shortest path $P'$ in $G$, and that every vertex on $P$ but not on $P'$ is adjacent to some vertex on $P'$. 
Note that pseudo-shortest paths are closed under concatenations, i.e. if we have a pseudo-shortest path from $u$ to $v$, and a pseudo-shortest path from $v$ to $w$, and a shortest path from $u$ to $w$ goes through $v$, then the concatenation of the paths is a pseudo-shortest path from $u$ to $w$.

\subparagraph*{Unit disk graphs} 
Given a set $S$ of points in $\R^2$,  
a unit disk graph $G = G(S) = (V, E)$ is a graph with vertices $V = S$, and edges between every pair of points $u,v\in S$ with distance less than $1$, i.e. $|u-v|\le 1$.
Note that by this definition, a unit disk has diameter $1$.
Observe that unit disk graphs have a natural embedding in the plane, albeit with potentially many crossing edges.
The following simple lemma about paths that cross the interior of a disk has been discovered in previous works (e.g.~\cite{YanXD12}); we include a proof here for the sake of completeness.
\begin{lemma} \label{lem:crossinginterior}
Let $G = (V, E)$ be a unit disk graph, and let $u,v,x \in V$ be three distinct vertices with the edge $uv \in E$ where the straight line edge defining $uv$ intersects the unit disk centered at $x$.
Then $x$ has an edge with at least one of $u$ or $v$.
\end{lemma}
\begin{proof}
Consider a point $p$ on $uv$ that lies in the unit disk centered at $x$. Since the edge $uv$ has length at most $1$, $p$ is also in the unit disk of either $u$ or $v$, so $x$ has an edge with either $u$ or $v$.
\end{proof}

\section{Path systems of pseudo-shortest paths in unit disk graphs}
\label{sec:pathsystem}

In this section, we prove the existence of a path system $\Pi$ of pseudo-shortest paths consisting of Delaunay triangulation edges for a unit disk graph $G$, and a fixed starting vertex $s$. We let $W_a$ denote the set of vertices at distance $a$ from a root vertex $s$.
We will show how to build the paths $\Pi[u]$ for a vertex $u\in W_a$ by first
a pseudo-shortest path to the nearest neighbor $w\in W_{a-1}$, then extending it to a pseudo-shortest path to $s$ from $w$.

\subsection{Nearest neighbors, Voronoi cells, and Delaunay paths}

Consider an edge $uw\in E$. Either $u$ and $w$ are in adjacent Voronoi cells, or the line segment $uw$ crosses a sequence of Voronoi cells of the points $u = v_1, v_2, \dots, v_\ell = w$ for some $\ell$. As the Voronoi diagram is dual to the Delaunay triangulation, between each pair of Voronoi cells there is a Delaunay edge. This induces a path in the Delaunay triangulation $\DT(G)$ between $u$ and $w$, which we call the {\em Delaunay path between $u$ and $w$}. Such paths were first considered by Dickerson and Drysdale~\cite{DickersonD90}, and also later by Cabello and \Jejcic~\cite{CabelloJ15}. 


\begin{lemma}[Dickerson-Drysdale~\cite{DickersonD90}; Cabello-\Jejcic~\cite{CabelloJ15}] \label{lem:DD}
Let $uw\in E$. Let $P =[u = v_1, v_2, \dots, v_\ell= w]$ be the Delaunay path between $u$ and $w$. Then the following holds:
\begin{enumerate}
    \item $|v_i-v_j|\leq 1$ for all $1 \leq i < j \leq \ell$, {\em i.e.} all pairs of vertices of $P$ are connected to each other in the unit disk graph.
    \item All vertices of $P$ lie inside the disk with $uw$ as diameter.
    \item For all $1\leq i<j<k \leq \ell$, we have $|v_i-v_j|<|v_i-v_k|$.
\end{enumerate}
\end{lemma}

In particular, if $u\in W_a$ and $w$ is the nearest%
\footnote{In this section, \emph{nearest} refers to distances in the Euclidean metric, not in the graph metric.}
neighbor of $u$ in $W_{a-1}$, then we denote by $\DP[u]$ the Delaunay path between $u$ and $w$, and call it the {\em Delaunay path for $u$}. In this case, it follows immediately from Lemma~\ref{lem:DD} that for all $i<\ell$, we have $v_i\in W_a$. Furthermore, $\DP[u]$ is a pseudo-shortest path from $u$ to $w$.




\subsection{Non-crossing property of Delaunay paths}
We now prove that the Delaunay paths do not cross.

\begin{lemma}\label{lem:delaunaypath_samelv}
    Let $u_1,u_2\in W_a$ and $w_1,w_2$ be the corresponding nearest neighbors in $W_{a-1}$.
    Then the Delaunay paths for $u_1$ and $u_2$ do not cross.
\end{lemma}
\begin{proof}
    \begin{figure}
        \centering
        \includegraphics[scale=0.8, page=2]{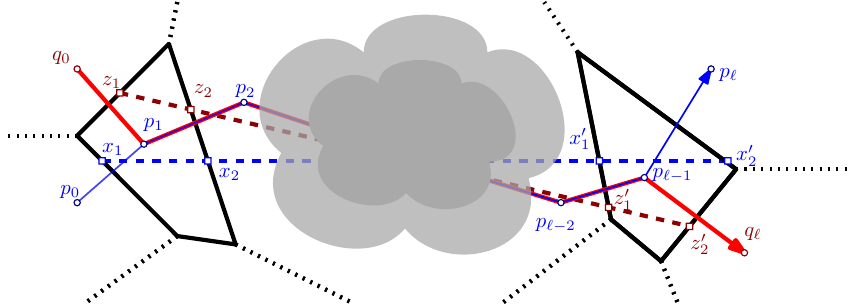}
        \caption{The points $x_1$, $z_1$, $z_2$, $x_2$ are clockwise around the Voronoi cell, leading to $z_1$ and $z_2$ having positive $y$-coordinates. The cloud covers the intermediate vertices on the Delaunay paths.}
        \label{fig:noncrossing_voronoi}
    \end{figure}
    See Figure \ref{fig:noncrossing_voronoi} throughout this proof. Consider the line segments $u_1w_1$ and $u_2w_2$.
    Note that these two line segments don't cross as if they did, then either $w_1$ would be closer to $u_2$ or $w_1$ would be closer to $u_2$.
    By rotating the plane and a suitable shift, 
    we may assume that $u_1w_1$ is horizontal line on the $x$ axis.
    Now let's consider the Delaunay paths $P=\DP[u_1]$ from $u_1$ to $w_1$ and $Q=\DP[u_2]$ from $u_2$ to $w_2$.
    We assume for sake of contradiction that $P$ and $Q$ cross.
    Since $P$ and $Q$ consists of edges in a Delaunay triangulation, there can be no simple crossings, thus there is some subpath $P' = [p_0, p_1, \dots, p_\ell]$ of $P$ and (possibly reversed) subpath $Q' = [q_0, q_1, \dots q_\ell]$ of $Q$ that form a forward crossing.

    Consider the Voronoi cell of $p = p_1 = q_1$. Let $x_1$ and $x_2$ be the intersection of the line segment $u_1w_1$ with the Voronoi cell $p$, and $z_1$ and $z_2$ be the intersection of the line segment $u_2w_2$ with the Voronoi cell of $p$. 
    Without loss of generality,
    we may assume the cyclic ordering $p_0, q_0, p_2$ around $p_1$ is clockwise 
    (or we can reflect everything about the $x$ axis).
    This implies that the Voronoi edges between $p_1$ and $p_0$, $p_1$ and $q_0$, $p_1$ and $p_2$ are also ordered clockwise about the Voronoi cell for $p_1$. 
    As $P'$ and $Q'$ are parts of Delaunay paths, 
    we know that $x_1$ lies on the Voronoi edge between $p_1$ and $p_0$,
    $z_1$ lies on the Voronoi edge between $p_1$ and $q_0$,
    and both $z_2$ and $x_2$ lie on the Voronoi edge between $x_2$ and $p_2$.
    Furthermore, since $x_1x_2$ and $z_1z_2$ don't intersect (as they are subsets of $u_1w_1$ and $u_2w_2$), 
    we can conclude that $x_1$, $z_1$, $z_2$, $x_2$
    is the cyclic ordering in the clockwise direction on the Voronoi cell of $p_1$.
    Since $x_1$ and $x_2$ are on the $x$ axis, we can thus conclude that
    $z_1$ and $z_2$ have positive $y$-value.

    Now let's consider the Voronoi cell of $p' = p_{\ell-1} = q_{\ell-1}$, and define $x_1', x_2', z_1', z_2'$ analogously for this cell as we did before for the Voronoi cell of $p$.
    Since $P'$ and $Q'$ form a forward crossing, this implies $p_\ell, q_\ell, p_{\ell-2}$ are oriented clockwise as well.
    By the same argument as before about cyclic orderings, 
    this implies that $z_1'$ and $z_2'$ must have a negative $y$ value.
    However, $z_2$ had a positive $y$-value and $z_1'$ had a negative $y$-value.
    This implies that the ray $\ray(z_1, z_2)$ from $z_1$ to $z_2$ and the ray $\ray(z_2', z_1')$ from $z_2'$ to $z_1'$ intersect the line between $x_1$ and $x_2'$. 
    By symmetry, we can also show that the ray $\ray(x_1, x_2)$ and the ray $\ray(x_2', x_1')$ intersects the line between $z_1$ and $z_2'$. 
    Thus we conclude that the line segment between $z_1$ and $z_2'$ intersect the line segment between $x_1$ and $x_2'$, a contradiction.
\end{proof}

We can also show that Delaunay paths between different levels do not cross either.
\begin{lemma}\label{lem:delaunaypath_difflv}
    Let $u_1\in W_a$ and $u_2\in W_b$ for some postive integers $a$ and $b$ with $a<b$.
    Then the Delaunay paths $\Delta[u_1]$ and $\Delta[u_2]$ do not cross.
\end{lemma}
\begin{proof}
    Let $w_1$ be the nearest neighbor of $u_1$ in $W_{a-1}$, and $w_2$ be the nearest neighbor of $u_2$ in $W_{b-1}$.
    We consider the Delaunay paths $\DP[u_1] = [p_1 = u_1, p_2, \dots, p_k = w_1]$ and $\DP[u_2] = [q_1 = u_2, q_2, \dots q_\ell = w_2]$.
    First observe that since $\DP[u_1]$ and $\DP[u_2]$ are Delaunay paths, they contain subsets of edges in the Delaunay triangulation. They can only have forward or backwards crossings.
    However, by \Cref{lem:DD}, $q_i$ for $i < \ell$ is on $W_b$ and $q_\ell$ is in $W_{b-1}$, 
    while all of the vertices of $\DP[u_1]$ are on level $a$ or level $a-1$.
    Since $a<b$, $\DP[u_1]$ and $\DP[u_2]$ can share at most one vertex in common, namely when $b-1 = a$ and $q_\ell$ is the common vertex.
    Since $q_\ell$ is an endpoint of $\DP[u_1]$, a forward (or backward) crossing between $\DP[u_1]$ and $\DP[u_2]$ is not possible.
\end{proof}

\subsection{Constructing pseudo-shortest paths from Delaunay paths}

\begin{figure}
    \centering
    \includegraphics[height=0.15\textheight]{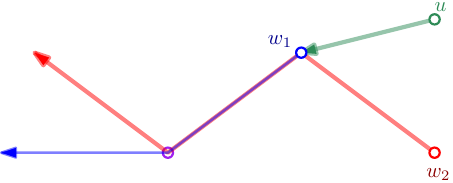}
    \caption{The paths $\Pi[w_1]$ (in blue) and $\Pi[w_2]$ (in red) for $w_1, w_2\in W_{a-1}$ don't cross.
    However, the Delaunay path $\DP[u]$ consisting of $u\in W_a$ to $w_1$ (in green) concatenated with $\Pi[w_1]$ forms a path that crosses $\Pi[w_2]$.}
    
    \label{fig:careless_attachments}
\end{figure}

We have shown via Lemma~\ref{lem:delaunaypath_samelv} and Lemma~\ref{lem:delaunaypath_difflv} that no two Delaunay paths cross.
This naturally suggests an intuitive way to inductively build $\Pi[u]$ based on distance from $s$. 
Suppose we have constructed a non-crossing path system $\Pi_{\le a-1}$ of pseudo-shortest paths of vertices at distance at most $a-1$ from $s$ by finding paths $\Pi[v]$ for all $v\in \bigcup_{i=0}^{a-1} W_{i}$, and we wanted to construct $\Pi[u]$ for $u\in W_a$. 
The natural method is to define $\Pi[u]$ to be the concatenation of the Delaunay path $\DP[u]$ that goes from $u$ to a vertex $w\in W_{a-1}$ with the path $\Pi[w]$. Clearly this is a pseudo-shortest path, as it is the concatenation of a pseudo-shortest path from $u$ to $w$ and from $w$ to $s$.

Unfortunately, carelessly extending paths in this manner may in fact create crossings! Consider two vertices $w_1, w_2 \in W_{a-1}$ that have non-crossing paths $\Pi[w_1]$ and $\Pi[w_2]$ where $w_1\in \Pi[w_2]$. If we attach the Delaunay path $\DP[u]$ for a vertex $u\in W_a$ that has nearest neighbor $w_1$, then this might induce a crossing as pictured in \Cref{fig:careless_attachments}.
Thankfully, with a little more care, we show that there is a way to extend the Delaunay path $\DP[u]$ for all $u\in W_a$ into full paths that don't cross each other or $\Pi_{\le a-1}$ with the following lemma.

\begin{lemma}[Path Extension Lemma] \label{lem:path_ext}
Let $G=(V,E)$ be a planar graph, and $\Pi$ be a non-crossing path system to $s\in V$ that contains a path $\Pi[w]$ from $w$ to $s$ for some $v\in V$. Suppose we had a non-crossing collection of paths $\mathcal{P}$ that end at $w$ where no path in $\mathcal{P}$ crosses any of the paths of $\Pi$. Then we can extend each path of $\mathcal{P}$ to end at $s$ without creating any crossing.
\end{lemma}

\begin{proof}
The extension proceeds as follows and is illustrated in \Cref{fig:snapping_down}: 
\begin{enumerate}
    \item  
    Consider a fixed perturbation to the paths of $\Pi$ and $\mathcal{P}$ such that none of the paths share any point in common, and all end points of $w$ lie in a small ball of radius $\varepsilon$ for an arbitrarily small $\varepsilon> 0$. This is possible because all paths are non-intersecting. See \cite[Section 4]{ChangEX15} for discussions of algorithms for constructing the perturbation.
    \item 
    For every path $P\in \mathcal{P}$, snap the end point of the path to the end point of the path of $\Pi$ that is clockwise around the $\varepsilon$-radius ball. We omit the formal description of this step, as it is more instructive to observe the illustration in \Cref{fig:snapping_down}.
    \item 
    Extend all paths $P\in \mathcal{P}$ by the forward continuation of the path of $\Pi$ we have snapped to.
\end{enumerate}
Observe that by construction, our extended (perturbed) paths are non-crossing. 
\end{proof}

This gives us a way to construct our path system $\Pi_{\le a-1}$ to $\Pi_{\le a}$ with Delaunay paths.
\begin{lemma}
\label{lem:path_ext_multiple}
Given a set of non-crossing pseudo-shortest paths $\Pi_{\le a-1}$,
we can extend each $\DP[u]$ for all $u\in W_a$ to a pseudo-shortest path $\Pi[u]$ in a way such that no two paths intersect each other or paths of $\Pi_{\le a-1}$.
\end{lemma}
\begin{proof}
\begin{figure}
    \centering
    \includegraphics[page=1, width=0.3\textwidth]{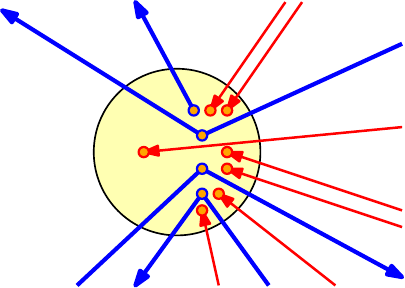}
    \includegraphics[page=2, width=0.3\textwidth]{figures/snapping_down.pdf}
    \includegraphics[page=3, width=0.3\textwidth]{figures/snapping_down.pdf}
    \caption{
    \textbf{(Left)} Step 1: The yellow disk represent the node $w$; blue paths represent paths in $\Pi$ that goes through $w$ ($\Pi^w$ in the case of Lemma~\ref{lem:path_ext_multiple}); red paths represent a non-crossing collection of paths that end at $w$ ($\Delta^w$ in the case of Lemma~\ref{lem:path_ext_multiple}).
    \textbf{(Middle)} Step 2: Snapping every red path to the vertex corresponding to the first blue path clockwise around the boundary.
    \textbf{(Right)} Step 3: Extending every red path by the continuation of the blue path marked in purple.
    }
    \label{fig:snapping_down}
\end{figure}
Fix a vertex $w\in W_{a-1}$ and consider the pseudo-shortest paths
\[ \Pi^{w} = \{\Pi[v] \in \Pi_{\le a-1} \mid \text{$\Pi[v]$ passes through $w$}\}. \]
Also consider the following Delaunay paths.
\[ \DP^w = \{\DP[u] \mid u\in W_a \text{ and the nearest neighbor of $u$ in $W_{a-1}$ is $w$}\}.\]
Note that no path of $\Pi^{w}$ crosses any path of $\DP^w$, and both are non-crossing path systems.

We repeatedly apply \cref{lem:path_ext} to each $w\in W_{a-1}$.
Doing so, we construct a path $\Pi[u]$ for each $u\in W_{a}$ that is crossing free with all paths of $\Pi_{\le a-1}$.
\end{proof}

By induction on distance from $s$ the lemma below follows.
\begin{lemma}[Spanning non-crossing path systems of psuedo-shortest paths]
\label{lem:delaunay-sncpspsp}
    Let $G$ be a unit disk graph on the point set $S$ and let $s$ be a fixed source vertex. 
    There exists a spanning non-crossing path system $\Pi$ of pseudo-shortest paths rooted at $s$ using only edges in the Delaunay triangulation of $S$.
\end{lemma}

\section{Shortest path separators in non-crossing spanning path systems}
\label{sec:pathseparator}

In this section we show how to find the shortest path separators on embedded 
triangulated planar graphs $G = (V,E)$
when given a spanning non-crossing path systems $\Pi$ to $s\in V$.
We aim to apply the lemma of Lipton and Tarjan albeit onto a triangulated graph $G'$ containing a slightly perturbed version of $\Pi$.
Formally we define a \emph{perturbation of a path system $\Pi$ rooted at $s$} as the collection of all paths with vertices except $s$ perturbed within a small ball of radius $\eps > 0$ such that none of the paths are crossing.
Observe that a single vertex $u \in V$ can be perturbed to many copies $u_1, ..., u_k$ corresponding to the same vertex $u\in V$.
See \Cref{fig:perturbed_pathsystem} for an example of this. 

The following lemma shows that this graph can be constructed while ensuring the additional edges we add are between vertices corresponding to edges in the original graph $G$.
\begin{lemma} \label{lem:perturb}
Given a triangulated embedded planar graph $G = (V,E)$, 
and a non-crossing spanning path system $\Pi$ to a vertex $s \in V$, 
there exists a planar graph $G' = (V', E')$ that is the triangulation of the perturbation of $\Pi$ such that every edge $e\in E'$ belongs to one of the following categories:
\begin{itemize}
    \item Edges $E_\path$ between vertices $u_i, v_j \in V'$ that correspond to $u, v \in V$ where $uv \in \Pi$.
    \item Edges $E_\vertex$ between $v_i, v_j \in V'$ that correspond to the same vertex $v \in V$.
    \item Edges $E_\triangulation$ between vertices $u_i, v_j \in V'$ that correspond to $u, v \in V$ where $uv \in E$.
\end{itemize}
\end{lemma}

\begin{figure}
    \centering
    \includegraphics[scale=0.5, page=3]{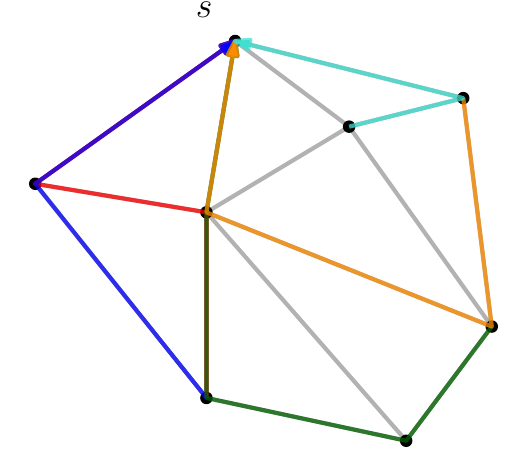}
    \includegraphics[scale=0.5, page=4]{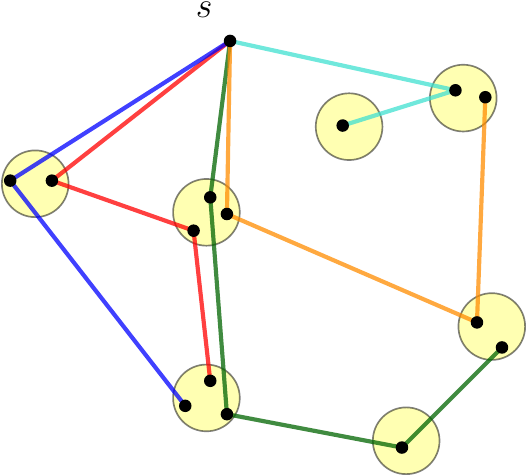}
    \includegraphics[scale=0.5, page=5]{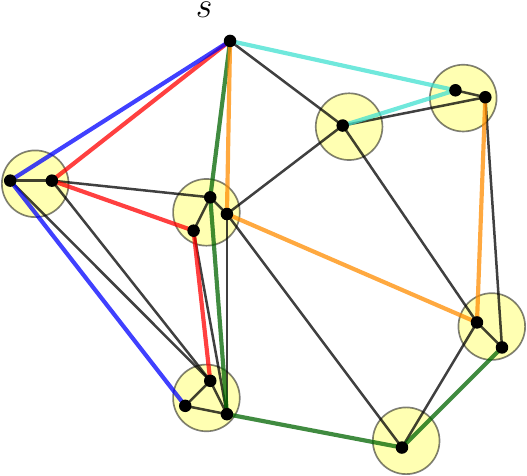}
    \caption{\textbf{(Left)} A path system $\Pi$ to a vertex $s$ in a Delaunay triangulation. 
    \textbf{(Middle)} A perturbation of all points on all paths at vertices other than $s$ by at most $\eps$ such that no two paths intersect except at $s$.
    \textbf{(Right)} A triangulation of the perturbation of the paths using only previously existing edges.
    }
    \label{fig:perturbed_pathsystem}
\end{figure}
\begin{proof}
We describe a construction of $G'$. Let $V'$ be the vertices of the perturbation of $\Pi$, consisting of all vertices $u\in V$ with $u\neq s$ that are perturbed to $Q_u = \{u_1, u_2, ..., u_{\ell(u)}\}$ where $\ell(u)$ is the number of paths using vertex $u$ in $\Pi$. 
We say an edge $u_i,v_j \in V'$ is \emph{faithful} if $u_i\in Q_u$ and $v_j\in Q_v$ and $uv \in G$.
Let $E_\path$ denote the edges that correspond to the perturbed paths of $\Pi$. 
Observe that the edges of $E_\path$ are faithful.

For all perturbations of the same vertex $u\in V$ to $Q_u = \{u_1, u_2, ..., u_{\ell(u)}\}\subset V'$,  add a maximal set of planar edges that do not cross the edges of $E_\path$ (i.e. a restricted triangulation in the local neighborhood of the perturbation), 
and let $E_\vertex^u$ denote these added edges. 
Define $E_\vertex = \bigcup_{u\in V} E_\vertex^u$. 

%

Now let's consider an edge $uv \in E$ with no path in $\Pi$ using $uv$. 
We claim that there exists a perturbed vertex $u_i \in Q_u$ and $v_j \in Q_v$ such that the edge $u_iv_j$ does not intersect edges of $E_\vertex \cup F$ where $F$ is any non-crossing collection of faithful edges. 
Thus it is possible to add a collection of non-crossing faithful edges $E_\triangulation^1$ that do not intersect $E_\vertex \cup E_\path$ so that every edge $uv\in E$ has a faithful edge.

Now consider the faces in the graph $(V', E_\path \cup E_\vertex \cup E_\triangulation^1)$, they are either: (1) a face that consists of three faithful edges corresponding to a triangle $uvw$ in $G$ joined by chains of vertices of $V'$ corresponding to vertices $u$, $v$, or $w$ (2) a face with two faithful edges corresponding to the same edge $uv\in E$.
Either case is easy to triangulate with an additional set of non-crossing faithful edges $E_\triangulation^2$ since any non-crossing triangulation can only consist of faithful edges. 
Let $E_\triangulation = E_\triangulation^1 \cup E_\triangulation^2$.
\end{proof}

With this lemma we present our result for any non-crossing spanning path system of a triangulated planar graph.

\begin{lemma}[Balanced separators for non-crossing path systems] \label{lem:separator}
Given a triangulated embedded planar graph $G = (V,E)$, and a non-crossing spanning path system $\Pi$ to a vertex $s \in V$, 
there exists a Jordan curve $C$ with the following properties:
\begin{enumerate}
    \item 
    There are at most $2n/3$ vertices in $V_\inside(C)$ and $V_\outside(C)$.
    \item $C$ is defined by two paths $P_u$ from $u$ to $s$ and $P_v$ from $v$ to $s$ and either one edge $uv \in E$ or $u=v$. 
    \item $P_u$ is the suffix of a path $\Pi[u']$ and $P_v$ is a suffix of a path $\Pi[v']$ for some $u', v' \in V$.
\end{enumerate}
\end{lemma}
\begin{proof}
Begin by constructing the triangulated perturbed graph $G'=(V', E')$ from $G$ and $\Pi$ with \Cref{lem:perturb}. 
Let $T$ be the tree defined by the perturbed paths of $\Pi$ in $G'$ rooted at $s$. 

Observe that $T$ is spanning and the fundamental cycles of $T$ contain $s$.
For vertices $u_1, ..., u_{\ell(u)} \in V'$ that correspond to the vertex $u\in V$, we  arbitrarily choose one vertex (say $u_1$) and give it weight $1$, and give all other vertices weight $0$. We do this for all vertices of $V$.
Now we can apply the weighted separator theorem of \Cref{thm:LTseparator} to $G'$ with this weight function, and spanning tree defined by the perturbed paths of $\Pi$
to get a balanced separator $C$ satisfying condition 1.
Observe that $C$ corresponds to a fundamental cycle of the perturbed paths,
which is an edge $u_iv_j\in E'$, and two paths that end at $s$ and correspond to paths of $\Pi$ satisfying condition 3.
Finally observe that the $u_i$ and $v_j$ correspond either to two vertices $u,v \in V$ with $uv\in E$ or to the same vertex $u\in V$. In either case condition 2 is satisfied.
\end{proof}

\section{Constructing the shortest path 1-neighborhood separator}
\label{sec:1hseparator}

\subsection{Delaunay edges are crominating}
Recall that a pair of vertices $u,v \in V$ is crominating if for all edges of $G$ intersecting the line segment between $u$ and $v$ has one end point adjacent to $u$ or $v$.
We will first show that Delaunay edges are crominating.
\begin{lemma} 
\label{lem:crominating}
For a unit disk graph $G=(V,E)$, all edges in the Delaunay triangulation $DT(G)$ are crominating.
\end{lemma}
\begin{proof}
For the sake of contradiction, suppose that there exists Delaunay edge $uv\in DT(G)$ and edge $xy\in E$, such that $uv$ crosses but does not dominate $xy$.
Without loss of generality, we may assume that the unit disk centered at $x$ intersects $uv$, $uv$ is horizontal, and $x$ lies above $uv$.
Consider if $|u-v| \le 1$ so that $uv \in E$. This would mean that $xy$ would need to intersect either the unit disk at $u$ or $v$. By Lemma~\ref{lem:crossinginterior}, this implies that one of $x$ or $y$ is adjacent to either $u$ or $v$, a contradiction. Thus we will focus on the case where $|u-v| > 1$.

Let $D^{\uparrow}_{uv}$ (resp. $D^\downarrow_{uv}$) be the semidisk above (resp. below) $uv$ with diameter $uv$.
Since $|x-u|,|x-v|>1$ and the distance between $x$ and line segment $uv$ is at most $1/2$, we have $x\in D^\uparrow_{uv}$.
Now, let $D_x$ be the disk with center $x$ and radius $1$, and let $D^\downarrow_x$ be the part of $D_x$ that lies below $uv$. Clearly $y\in D^\downarrow_x$. Since $|x-u|,|x-v|>1$, it follows that $D^\downarrow_x\subset D^\downarrow_{uv}$, and thus $y\in D^\downarrow_{uv}$. (See Figure~\ref{fig:delaunay-crominating-1}.)

Let $z\in S$ be the point above $uv$ such that $u,z,v$ share the same face in $DT(S)$. Let $D_{uvz}$ be the circumcircle through $u,v,z$. There are two cases to consider.
\begin{enumerate}
    \item If $z\in D^\uparrow_{uv}$, then $D_{uvz}\supset D^\downarrow_{uv}$, and thus must contain $y$. However, this is forbidden by the definition of Delaunay triangulation, a contradiction. 
    \item If $z\notin D^\uparrow_{uv}$, then $D_{uvz}\supset D^\uparrow_{uv}$, and thus must contain $x$, which is a contradiction. \qedhere
\end{enumerate}
\end{proof}

\begin{figure}
    \centering
    \includegraphics[scale = 0.9, page=1]{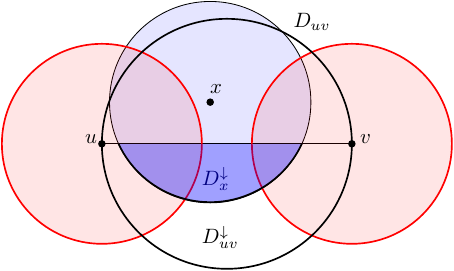}
    \caption{Suppose an edge $xy\in E$ intersects but is not dominated by a Delaunay edge $uv$. Without loss of generality assume that disk $x$ intersects edge $uv$ and lies above $uv$. We show that $x\in D^\uparrow_{uv}$ and $y\in D^\downarrow_x\subset D^\downarrow_{uv}$.}
    \label{fig:delaunay-crominating-1}
\end{figure}

\subsection{Putting everything together}

We finally have all the pieces we need to prove that unit disk graphs have a shortest path $1$-neighborhood separator. 

Let $G = (V, E)$ be a unit disk graph on $n$ points.
Fix an arbitrary source vertex $s \in V$.
By \Cref{lem:delaunay-sncpspsp}, there exists a spanning non-crossing path system $\Pi$ of pseudo-shortest paths rooted at $s$ using only the edges in the triangulation of $\DT(G)$\footnote{Note that $\DT(G)$ may not be triangulated as a planar graph, the outer face is typically not a triangle.}. 
Let $C$ be a Jordan curve according to Lemma~\ref{lem:separator}, and let $P_u,P_v$ be the paths defining $C$ as in Lemma~\ref{lem:separator}. Consider an edge $xy\in E$ that intersects $C$. There are a few cases:
\begin{enumerate}
    \item $xy$ intersects the Delaunay edge $uv$.
    \item $xy$ intersects an edge $wz$ of the triangulated outer face.
    \item $xy$ intersects an edge $wz$ in either $\Pi[u]$ or $\Pi[v]$.
\end{enumerate}
In the first case, $xy$ is crominated by $uv$, thus, removing the 1-neighborhood of $u$ and $v$ removes the edge $xy$.
In the second case, note that the edge $wz$ in the triangulated outer face is outside of the Delaunay triangulation and thus outside the convex hull of the point set, so it is impossible for an edge $xy\in E$ to cross.
In the third case, assume without loss of generality that $wz\in\Pi[u]$. 
The edge $wz$ is part of a Delaunay path between two vertices $w^*$ and $z^*$ which are two vertices on the shortest path from $u$ to $s$. 
By Lemma~\ref{lem:DD}, this entire edge is covered by the disk centered at $w^*$ and the disk centered at $z^*$. In particular that means that $xy$ intersects at least one of the disks centered at $w^*$ or $z^*$, and thus applying \Cref{lem:crossinginterior} shows that either $x$ or $y$ is a neighbor of $w^*$ or $z^*$.
We can now conclude that removing the $1$-neighborhood of $u,v,P_u$ and $P_v$ removes all edges that cross $C$.

\begin{theorem}
Every unit disk graph admits a shortest path $1$-neighborhood separator.
\end{theorem}

\subparagraph*{Remarks} We have omitted discussions of runtimes related to the construction of our $1$-neighborhood separator, but we note that we can construct the separator in $O(n^2)$ time.
Finding the path system takes $O(n^2)$ time. Indeed, the path system has $O(n^2)$ size, although more compact representations of the path system are possible. Furthermore, it can be shown that performing the perturbation can be done in $O(n^2\log n)$ time by the algorithm of \cite{ChangEX15}, and the separator construction of \Cref{lem:separator} can be done in $O(n^2)$. 
We leave as an open problem as to whether it is possible to construct the separator faster.
To contrast, finding shortest path separators in planar graphs can be done in $O(n)$ time using the dual co-tree.

\bibliographystyle{plainurl}
\bibliography{reference}

\begin{thebibliography}{10}

\bibitem{AbrahamG06}
Ittai Abraham and Cyril Gavoille.
\newblock Object location using path separators.
\newblock In Eric Ruppert and Dahlia Malkhi, editors, {\em Proceedings of the
  Twenty-Fifth Annual {ACM} Symposium on Principles of Distributed Computing,
  {PODC} 2006, Denver, CO, USA, July 23-26, 2006}, pages 188--197. {ACM}, 2006.
\newblock \href {https://doi.org/10.1145/1146381.1146411}
  {\path{doi:10.1145/1146381.1146411}}.

\bibitem{AlonSR90}
Noga Alon, Paul Seymour, and Robin Thomas.
\newblock A separator theorem for nonplanar graphs.
\newblock {\em J. Amer. Math. Soc.}, 3(4):801--808, 1990.
\newblock \href {https://doi.org/10.2307/1990903} {\path{doi:10.2307/1990903}}.

\bibitem{CabelloJ15}
Sergio Cabello and Miha Jej{\v{c}}i{\v{c}}.
\newblock Shortest paths in intersection graphs of unit disks.
\newblock {\em Comput. Geom.}, 48(4):360--367, 2015.
\newblock \href {https://doi.org/10.1016/j.comgeo.2014.12.003}
  {\path{doi:10.1016/j.comgeo.2014.12.003}}.

\bibitem{CarmiCKKORRSS20}
Paz Carmi, Man{-}Kwun Chiu, Matthew~J. Katz, Matias Korman, Yoshio Okamoto,
  Andr{\'{e}} van Renssen, Marcel Roeloffzen, Taichi Shiitada, and Shakhar
  Smorodinsky.
\newblock Balanced line separators of unit disk graphs.
\newblock {\em Comput. Geom.}, 86, 2020.
\newblock \href {https://doi.org/10.1016/j.comgeo.2019.101575}
  {\path{doi:10.1016/j.comgeo.2019.101575}}.

\bibitem{Chan03}
Timothy~M. Chan.
\newblock Polynomial-time approximation schemes for packing and piercing fat
  objects.
\newblock {\em J. Algorithms}, 46(2):178--189, 2003.
\newblock \href {https://doi.org/10.1016/S0196-6774(02)00294-8}
  {\path{doi:10.1016/S0196-6774(02)00294-8}}.

\bibitem{Chan23}
Timothy~M. Chan.
\newblock Finding triangles and other small subgraphs in geometric intersection
  graphs.
\newblock In Nikhil Bansal and Viswanath Nagarajan, editors, {\em Proceedings
  of the 2023 {ACM-SIAM} Symposium on Discrete Algorithms, {SODA} 2023,
  Florence, Italy, January 22-25, 2023}, pages 1777--1805. {SIAM}, 2023.
\newblock URL: \url{https://doi.org/10.1137/1.9781611977554.ch68}, \href
  {https://doi.org/10.1137/1.9781611977554.CH68}
  {\path{doi:10.1137/1.9781611977554.CH68}}.

\bibitem{ChanS19}
Timothy~M. Chan and Dimitrios Skrepetos.
\newblock All-pairs shortest paths in geometric intersection graphs.
\newblock {\em J. Comput. Geom.}, 10(1):27--41, 2019.
\newblock URL: \url{https://doi.org/10.20382/jocg.v10i1a2}, \href
  {https://doi.org/10.20382/JOCG.V10I1A2} {\path{doi:10.20382/JOCG.V10I1A2}}.

\bibitem{ChanS19a}
Timothy~M. Chan and Dimitrios Skrepetos.
\newblock Approximate shortest paths and distance oracles in weighted unit-disk
  graphs.
\newblock {\em J. Comput. Geom.}, 10(2):3--20, 2019.
\newblock \href {https://doi.org/10.20382/JOCG.V10I2A2}
  {\path{doi:10.20382/JOCG.V10I2A2}}.

\bibitem{ChangEX15}
Hsien{-}Chih Chang, Jeff Erickson, and Chao Xu.
\newblock Detecting weakly simple polygons.
\newblock In Piotr Indyk, editor, {\em Proceedings of the Twenty-Sixth Annual
  {ACM-SIAM} Symposium on Discrete Algorithms, {SODA} 2015, San Diego, CA, USA,
  January 4-6, 2015}, pages 1655--1670. {SIAM}, 2015.
\newblock \href {https://doi.org/10.1137/1.9781611973730.110}
  {\path{doi:10.1137/1.9781611973730.110}}.

\bibitem{ChangGL24}
Hsien{-}Chih Chang, Jie Gao, and Hung Le.
\newblock Computing diameter+2 in truly-subquadratic time for unit-disk graphs.
\newblock In Wolfgang Mulzer and Jeff~M. Phillips, editors, {\em 40th
  International Symposium on Computational Geometry, SoCG 2024, June 11-14,
  2024, Athens, Greece}, volume 293 of {\em LIPIcs}, pages 38:1--38:14. Schloss
  Dagstuhl - Leibniz-Zentrum f{\"{u}}r Informatik, 2024.
\newblock URL: \url{https://doi.org/10.4230/LIPIcs.SoCG.2024.38}, \href
  {https://doi.org/10.4230/LIPICS.SOCG.2024.38}
  {\path{doi:10.4230/LIPICS.SOCG.2024.38}}.

\bibitem{ChangKT22}
Hsien{-}Chih Chang, Robert Krauthgamer, and Zihan Tan.
\newblock Almost-linear \emph{{\(\epsilon\)}}-emulators for planar graphs.
\newblock In Stefano Leonardi and Anupam Gupta, editors, {\em {STOC} '22: 54th
  Annual {ACM} {SIGACT} Symposium on Theory of Computing, Rome, Italy, June 20
  - 24, 2022}, pages 1311--1324. {ACM}, 2022.
\newblock \href {https://doi.org/10.1145/3519935.3519998}
  {\path{doi:10.1145/3519935.3519998}}.

\bibitem{Chew89}
L.~Paul Chew.
\newblock Constrained delaunay triangulations.
\newblock {\em Algorithmica}, 4(1):97--108, 1989.
\newblock \href {https://doi.org/10.1007/BF01553881}
  {\path{doi:10.1007/BF01553881}}.

\bibitem{ClarkCJ90}
Brent~N. Clark, Charles~J. Colbourn, and David~S. Johnson.
\newblock Unit disk graphs.
\newblock {\em Discret. Math.}, 86(1-3):165--177, 1990.
\newblock \href {https://doi.org/10.1016/0012-365X(90)90358-O}
  {\path{doi:10.1016/0012-365X(90)90358-O}}.

\bibitem{BergKMT23}
Mark de~Berg, S{\'{a}}ndor Kisfaludi{-}Bak, Morteza Monemizadeh, and Leonidas
  Theocharous.
\newblock Clique-based separators for geometric intersection graphs.
\newblock {\em Algorithmica}, 85(6):1652--1678, 2023.
\newblock \href {https://doi.org/10.1007/S00453-022-01041-8}
  {\path{doi:10.1007/S00453-022-01041-8}}.

\bibitem{DickersonD90}
Matthew Dickerson and Robert L. (Scot)~Drysdale III.
\newblock Fixed-radius near neighbors search algorithms for points and
  segments.
\newblock {\em Inf. Process. Lett.}, 35(5):269--273, 1990.
\newblock \href {https://doi.org/10.1016/0020-0190(90)90056-4}
  {\path{doi:10.1016/0020-0190(90)90056-4}}.

\bibitem{djidjev1981separator}
Hristo~N Djidjev.
\newblock A separator theorem.
\newblock {\em Compt. Rend. Acad. Bulg. Sci.}, 34(5):643--645, 1981.

\bibitem{EricksonN11}
Jeff Erickson and Amir Nayyeri.
\newblock Shortest non-crossing walks in the plane.
\newblock In Dana Randall, editor, {\em Proceedings of the Twenty-Second Annual
  {ACM-SIAM} Symposium on Discrete Algorithms, {SODA} 2011, San Francisco,
  California, USA, January 23-25, 2011}, pages 297--208. {SIAM}, 2011.
\newblock \href {https://doi.org/10.1137/1.9781611973082.25}
  {\path{doi:10.1137/1.9781611973082.25}}.

\bibitem{ErlebachJS05}
Thomas Erlebach, Klaus Jansen, and Eike Seidel.
\newblock Polynomial-time approximation schemes for geometric intersection
  graphs.
\newblock {\em {SIAM} J. Comput.}, 34(6):1302--1323, 2005.
\newblock \href {https://doi.org/10.1137/S0097539702402676}
  {\path{doi:10.1137/S0097539702402676}}.

\bibitem{FoxPach08}
Jacob Fox and J\'{a}nos Pach.
\newblock Separator theorems and {T}ur\'{a}n-type results for planar
  intersection graphs.
\newblock {\em Adv. Math.}, 219(3):1070--1080, 2008.
\newblock \href {https://doi.org/10.1016/j.aim.2008.06.002}
  {\path{doi:10.1016/j.aim.2008.06.002}}.

\bibitem{FoxP14}
Jacob Fox and J{\'{a}}nos Pach.
\newblock Applications of a new separator theorem for string graphs.
\newblock {\em Comb. Probab. Comput.}, 23(1):66--74, 2014.
\newblock \href {https://doi.org/10.1017/S0963548313000412}
  {\path{doi:10.1017/S0963548313000412}}.

\bibitem{Frederickson87}
Greg~N. Frederickson.
\newblock Fast algorithms for shortest paths in planar graphs, with
  applications.
\newblock {\em {SIAM} J. Comput.}, 16(6):1004--1022, 1987.
\newblock \href {https://doi.org/10.1137/0216064} {\path{doi:10.1137/0216064}}.

\bibitem{FriggstadM23}
Zachary Friggstad and Ramin Mousavi.
\newblock An {O}(log k)-approximation for directed steiner tree in planar
  graphs.
\newblock In Kousha Etessami, Uriel Feige, and Gabriele Puppis, editors, {\em
  50th International Colloquium on Automata, Languages, and Programming,
  {ICALP} 2023, July 10-14, 2023, Paderborn, Germany}, volume 261 of {\em
  LIPIcs}, pages 63:1--63:14. Schloss Dagstuhl - Leibniz-Zentrum f{\"{u}}r
  Informatik, 2023.
\newblock \href {https://doi.org/10.4230/LIPICS.ICALP.2023.63}
  {\path{doi:10.4230/LIPICS.ICALP.2023.63}}.

\bibitem{GaoZ05}
Jie Gao and Li~Zhang.
\newblock Well-separated pair decomposition for the unit-disk graph metric and
  its applications.
\newblock {\em {SIAM} J. Comput.}, 35(1):151--169, 2005.
\newblock \href {https://doi.org/10.1137/S0097539703436357}
  {\path{doi:10.1137/S0097539703436357}}.

\bibitem{GilbertHT84}
John~R. Gilbert, Joan~P. Hutchinson, and Robert~Endre Tarjan.
\newblock A separator theorem for graphs of bounded genus.
\newblock {\em J. Algorithms}, 5(3):391--407, 1984.
\newblock \href {https://doi.org/10.1016/0196-6774(84)90019-1}
  {\path{doi:10.1016/0196-6774(84)90019-1}}.

\bibitem{Goodrich95}
Michael~T. Goodrich.
\newblock Planar separators and parallel polygon triangulation.
\newblock {\em J. Comput. Syst. Sci.}, 51(3):374--389, 1995.
\newblock \href {https://doi.org/10.1006/JCSS.1995.1076}
  {\path{doi:10.1006/JCSS.1995.1076}}.

\bibitem{Har-PeledQ17}
Sariel Har{-}Peled and Kent Quanrud.
\newblock Approximation algorithms for polynomial-expansion and low-density
  graphs.
\newblock {\em {SIAM} J. Comput.}, 46(6):1712--1744, 2017.
\newblock \href {https://doi.org/10.1137/16M1079336}
  {\path{doi:10.1137/16M1079336}}.

\bibitem{Har-PeledY22}
Sariel Har{-}Peled and Everett Yang.
\newblock Approximation algorithms for maximum matchings in geometric
  intersection graphs.
\newblock In Xavier Goaoc and Michael Kerber, editors, {\em 38th International
  Symposium on Computational Geometry, SoCG 2022, June 7-10, 2022, Berlin,
  Germany}, volume 224 of {\em LIPIcs}, pages 47:1--47:13. Schloss Dagstuhl -
  Leibniz-Zentrum f{\"{u}}r Informatik, 2022.
\newblock \href {https://doi.org/10.4230/LIPICS.SOCG.2022.47}
  {\path{doi:10.4230/LIPICS.SOCG.2022.47}}.

\bibitem{HuntMRRRS98}
Harry B.~Hunt III, Madhav~V. Marathe, Venkatesh Radhakrishnan, S.~S. Ravi,
  Daniel~J. Rosenkrantz, and Richard~Edwin Stearns.
\newblock {NC}-approximation schemes for {NP-} and {PSPACE}-hard problems for
  geometric graphs.
\newblock {\em J. Algorithms}, 26(2):238--274, 1998.
\newblock \href {https://doi.org/10.1006/JAGM.1997.0903}
  {\path{doi:10.1006/JAGM.1997.0903}}.

\bibitem{KaplanMRS18}
Haim Kaplan, Wolfgang Mulzer, Liam Roditty, and Paul Seiferth.
\newblock Routing in unit disk graphs.
\newblock {\em Algorithmica}, 80(3):830--848, 2018.
\newblock \href {https://doi.org/10.1007/S00453-017-0308-2}
  {\path{doi:10.1007/S00453-017-0308-2}}.

\bibitem{KawarabayashiR10}
Ken{-}ichi Kawarabayashi and Bruce~A. Reed.
\newblock A separator theorem in minor-closed classes.
\newblock In {\em 51th Annual {IEEE} Symposium on Foundations of Computer
  Science, {FOCS} 2010, October 23-26, 2010, Las Vegas, Nevada, {USA}}, pages
  153--162. {IEEE} Computer Society, 2010.
\newblock \href {https://doi.org/10.1109/FOCS.2010.22}
  {\path{doi:10.1109/FOCS.2010.22}}.

\bibitem{KawarabayashiS21}
Ken{-}ichi Kawarabayashi and Anastasios Sidiropoulos.
\newblock Embeddings of planar quasimetrics into directed {\(\mathscr{l}\)}1
  and polylogarithmic approximation for directed sparsest-cut.
\newblock In {\em 62nd {IEEE} Annual Symposium on Foundations of Computer
  Science, {FOCS} 2021, Denver, CO, USA, February 7-10, 2022}, pages 480--491.
  {IEEE}, 2021.
\newblock \href {https://doi.org/10.1109/FOCS52979.2021.00055}
  {\path{doi:10.1109/FOCS52979.2021.00055}}.

\bibitem{Klein02}
Philip~N. Klein.
\newblock Preprocessing an undirected planar network to enable fast approximate
  distance queries.
\newblock In David Eppstein, editor, {\em Proceedings of the Thirteenth Annual
  {ACM-SIAM} Symposium on Discrete Algorithms, January 6-8, 2002, San
  Francisco, CA, {USA}}, pages 820--827. {ACM/SIAM}, 2002.
\newblock URL: \url{http://dl.acm.org/citation.cfm?id=545381.545488}.

\bibitem{KleinMS13}
Philip~N. Klein, Shay Mozes, and Christian Sommer.
\newblock Structured recursive separator decompositions for planar graphs in
  linear time.
\newblock In Dan Boneh, Tim Roughgarden, and Joan Feigenbaum, editors, {\em
  Symposium on Theory of Computing Conference, STOC'13, Palo Alto, CA, USA,
  June 1-4, 2013}, pages 505--514. {ACM}, 2013.
\newblock \href {https://doi.org/10.1145/2488608.2488672}
  {\path{doi:10.1145/2488608.2488672}}.

\bibitem{LeW21}
Hung Le and Christian Wulff{-}Nilsen.
\newblock Optimal approximate distance oracle for planar graphs.
\newblock In {\em 62nd {IEEE} Annual Symposium on Foundations of Computer
  Science, {FOCS} 2021, Denver, CO, USA, February 7-10, 2022}, pages 363--374.
  {IEEE}, 2021.
\newblock \href {https://doi.org/10.1109/FOCS52979.2021.00044}
  {\path{doi:10.1109/FOCS52979.2021.00044}}.

\bibitem{Lee17}
James~R. Lee.
\newblock Separators in region intersection graphs.
\newblock In Christos~H. Papadimitriou, editor, {\em 8th Innovations in
  Theoretical Computer Science Conference, {ITCS} 2017, January 9-11, 2017,
  Berkeley, CA, {USA}}, volume~67 of {\em LIPIcs}, pages 1:1--1:8. Schloss
  Dagstuhl - Leibniz-Zentrum f{\"{u}}r Informatik, 2017.
\newblock \href {https://doi.org/10.4230/LIPIcs.ITCS.2017.1}
  {\path{doi:10.4230/LIPIcs.ITCS.2017.1}}.

\bibitem{LiP19}
Jason Li and Merav Parter.
\newblock Planar diameter via metric compression.
\newblock In Moses Charikar and Edith Cohen, editors, {\em Proceedings of the
  51st Annual {ACM} {SIGACT} Symposium on Theory of Computing, {STOC} 2019,
  Phoenix, AZ, USA, June 23-26, 2019}, pages 152--163. {ACM}, 2019.
\newblock \href {https://doi.org/10.1145/3313276.3316358}
  {\path{doi:10.1145/3313276.3316358}}.

\bibitem{LiptonT1979}
Richard~J Lipton and Robert~Endre Tarjan.
\newblock A separator theorem for planar graphs.
\newblock {\em SIAM Journal on Applied Mathematics}, 36(2):177--189, 1979.
\newblock \href {https://doi.org/https://doi.org/10.1137/013601}
  {\path{doi:https://doi.org/10.1137/013601}}.

\bibitem{Matousek14}
Jir{\'{\i}} Matousek.
\newblock Near-optimal separators in string graphs.
\newblock {\em Comb. Probab. Comput.}, 23(1):135--139, 2014.
\newblock \href {https://doi.org/10.1017/S0963548313000400}
  {\path{doi:10.1017/S0963548313000400}}.

\bibitem{Miller86}
Gary~L. Miller.
\newblock Finding small simple cycle separators for 2-connected planar graphs.
\newblock {\em J. Comput. Syst. Sci.}, 32(3):265--279, 1986.
\newblock \href {https://doi.org/10.1016/0022-0000(86)90030-9}
  {\path{doi:10.1016/0022-0000(86)90030-9}}.

\bibitem{MillerTTV97}
Gary~L. Miller, Shang{-}Hua Teng, William~P. Thurston, and Stephen~A. Vavasis.
\newblock Separators for sphere-packings and nearest neighbor graphs.
\newblock {\em J. {ACM}}, 44(1):1--29, 1997.
\newblock \href {https://doi.org/10.1145/256292.256294}
  {\path{doi:10.1145/256292.256294}}.

\bibitem{SmithW98}
Warren~D. Smith and Nicholas~C. Wormald.
\newblock Geometric separator theorems {\&} applications.
\newblock In {\em 39th Annual Symposium on Foundations of Computer Science,
  {FOCS} '98, November 8-11, 1998, Palo Alto, California, {USA}}, pages
  232--243. {IEEE} Computer Society, 1998.
\newblock \href {https://doi.org/10.1109/SFCS.1998.743449}
  {\path{doi:10.1109/SFCS.1998.743449}}.

\bibitem{Thorup04}
Mikkel Thorup.
\newblock Compact oracles for reachability and approximate distances in planar
  digraphs.
\newblock {\em J. {ACM}}, 51(6):993--1024, 2004.
\newblock \href {https://doi.org/10.1145/1039488.1039493}
  {\path{doi:10.1145/1039488.1039493}}.

\bibitem{Wulff-Nilsen11}
Christian Wulff{-}Nilsen.
\newblock Separator theorems for minor-free and shallow minor-free graphs with
  applications.
\newblock In Rafail Ostrovsky, editor, {\em {IEEE} 52nd Annual Symposium on
  Foundations of Computer Science, {FOCS} 2011, Palm Springs, CA, USA, October
  22-25, 2011}, pages 37--46. {IEEE} Computer Society, 2011.
\newblock \href {https://doi.org/10.1109/FOCS.2011.15}
  {\path{doi:10.1109/FOCS.2011.15}}.

\bibitem{YanXD12}
Chenyu Yan, Yang Xiang, and Feodor~F. Dragan.
\newblock Compact and low delay routing labeling scheme for unit disk graphs.
\newblock {\em Comput. Geom.}, 45(7):305--325, 2012.
\newblock \href {https://doi.org/10.1016/J.COMGEO.2012.01.015}
  {\path{doi:10.1016/J.COMGEO.2012.01.015}}.

\end{thebibliography}

\end{document}